\documentclass[a4paper,USenglish,cleveref]{lipics-v2021}

\usepackage{suffix}
\usepackage{xstring}
\DeclareRobustCommand\scases[2]{\ifmmode #1\begin{cases}#2\end{cases}
	\else \[#1\begin{cases}#2\end{cases}\]
	\fi }
\WithSuffix\DeclareRobustCommand\scases*[2]{\saveexpandmode\expandarg \StrSubstitute{\noexpand#2}{&}{\ }[\vv]\StrSubstitute{\vv}{\empty\\}{\ #1}[\vvv]\restoreexpandmode \ensuremath{#1\ \vvv}}

\usepackage{ifthen}
\newboolean{isJournal}
\setboolean{isJournal}{true}
\ifthenelse{\boolean{isJournal}}{\newcommand{\mycases}[2]{\scases{#1}{#2}}
    \def\Jvar{}
    \def\Jfigs{}
    \newcommand{\J}[2][]{#2}
    \newcommand{\JO}[2][]{#2}
    
    \newcommand{\JFigure}[2][]{#2}
    
    \newenvironment{Jfigure}[1][]{\begin{figure}[#1]}{\end{figure}}
    \newenvironment{Jwrapfigure}[3][]{\begin{wrapfigure}[#1]{#2}{#3}}{\end{wrapfigure}}

}{\newcommand{\mycases}[2]{\scases*{#1}{#2}}
    \usepackage{etoolbox}
    \def\Jvar{}
    \def\Jfigs{}
\newcommand{\J}[2][]{\appto\Jvar{#1#2}}
    \newcommand{\JO}[2][]{#1}
    
    \newcommand{\JFigure}[2][]{\appto\Jfigs{#1#2}}

}
\JO{\hideLIPIcs}

\newcommand*{\block}[1]{\subparagraph{#1}}

\title{NP-Completeness for the Space-Optimality of Double-Array Tries}

\author{Hideo Bannai}{Tokyo Medical and Dental University, Japan}{hdbn.dsc@tmd.ac.jp}{https://orcid.org/0000-0002-6856-5185
 }{Supported by JSPS KAKENHI Grant Number JP20H04141}

 \author{Keisuke Goto}{LegalOn Technologies}{keisukegotou@gmail.com}{https://orcid.org/0000-0001-6964-6182}{}
 
 \author{Shunsuke Kanda}{LegalOn Technologies}{shnsk.knd@gmail.com}{https://orcid.org/0000-0002-5462-122X}{}

 \author{Dominik K\"oppl}{University of Yamanashi, Japan}{dkppl@yamanashi.ac.jp}{https://orcidid.org/0000-0002-8721-4444}{Supported by JSPS KAKENHI Grant Numbers JP21H05847 and JP21K17701.}

 \authorrunning{H. Bannai, K. Goto, S. Kanda, and D. K\"oppl } 

\Copyright{Hideo Bannai, Keisuke Goto, Shunsuke Kanda, and Dominik K\"oppl} 



\category{} 

\relatedversion{}

\JO{\nolinenumbers} 

\EventEditors{John Q. Open and Joan R. Access}
\EventNoEds{2}
\EventLongTitle{42nd Conference on Very Important Topics (CVIT 2016)}
\EventShortTitle{CVIT 2016}
\EventAcronym{CVIT}
\EventYear{2016}
\EventDate{December 24--27, 2016}
\EventLocation{Little Whinging, United Kingdom}
\EventLogo{}
\SeriesVolume{42}
\ArticleNo{23}

\usepackage[nointegrals]{wasysym}
\usepackage{etoolbox}
\usepackage[T1]{fontenc}
\usepackage{amsmath}
\usepackage{amsthm}
\usepackage{graphicx}
\usepackage[whole]{bxcjkjatype}
\usepackage{url}
\newtheorem{problem}{Problem}

\usepackage{xcolor}

\usepackage{wasysym}
\usepackage[export]{adjustbox}

\newcommand{\SODA}{\textsc{Soda}}

\newcommand{\node}[1]{n_{#1}}
\newcommand{\carr}{\mathit{check}}
\newcommand{\barr}{\mathit{base}}
\newcommand{\None}{\mathit{None}}

\begin{document}
\maketitle
\begin{abstract}
Indexing a set of strings for prefix search or membership queries is a fundamental task with many applications such as information retrieval or database systems.
A classic abstract data type for modelling such an index is a trie.
Due to the fundamental nature of this problem, it has sparked much interest, leading to a variety of trie implementations with different characteristics.
A trie implementation that has been well-used in practice is the double-array (trie) consisting of merely two integer arrays.
While a traversal takes constant time per node visit, the needed space consumption in computer words can be as large as the product of the number of nodes and the alphabet size.
Despite that several heuristics have been proposed on lowering the space requirements, we are unaware of any theoretical guarantees.

In this paper, we study the decision problem whether there exists a double-array of a given size.
To this end, we first draw a connection to the sparse matrix compression problem, which makes our problem NP-complete for alphabet sizes linear to the number of nodes.
We further propose a reduction from the restricted directed Hamiltonian path problem,
leading to NP-completeness even for logarithmic-sized alphabets.
\end{abstract}
\JO[\newpage]{}

\section{Introduction}
A trie \cite{fredkin1960trie} is an edge-labeled tree structure for storing a set of strings and retrieving them.
The \emph{double-array}~\cite{aoe1989efficient} is a trie representation that consists of just two arrays.
It has been empirically observed that double-arrays exhibit a good balance between operational time and storage.
Thanks to this virtue,
the usage of the double-array can be discovered in a wide range of applications such as string dictionaries~\cite{kanda2017compressed,yata2007compact}, tokenization~\cite{song2021fast,yoshinaga2023back}, pattern matching~\cite{kanda2023engineering}, language models~\cite{norimatsu2016fast} and text classification~\cite{yoshinaga2014self}.

The double-array represents a trie using two arrays called $\barr$ and $\carr$ of equal length.
Nodes are assigned to their elements, and edges are represented using their values.
The values are determined to satisfy the conditions of the double-array.
In other words, the values can be freely determined as long as the conditions are satisfied.
There is an infinite number of layouts of the double-array (i.e., $\barr$ and $\carr$) that represent some trie, and the size of the double-array (i.e., the length of $\barr$ and $\carr$) varies according to the layout.
Smaller layouts allow us to shorten the sizes of both arrays, and thus help us to obtain a memory-saving double-array representation.

In practice, we usually determine the values of $\barr$ and $\carr$ with a greedy strategy:
visiting nodes in the depth- or breadth-first order and searching for valid values from the front of the arrays.
This computes a layout that minimizes the size when adding a new node from the current state.
The greedy search is simple to implement and often works well for strings with a small alphabet such as DNA or ASCII\@.
This is because it is easy to search for values that satisfy the conditions.
However, it is locally optimal, and there is no guarantee that we obtain a double-array whose space is optimal.
It has been empirically observed that the size of the double-array grows  super-linearly with the alphabet size (for large enough alphabets),
such that memory efficiency has become an issue \cite{kanda2023engineering,liu2011compression,norimatsu2016fast};
these empirical evaluations sparked questions on whether there exist bounds on how good the greedy approach approximates the optimal space, or whether we can hope for another approach that gives optimal space guarantees.

Given that our trie stores $M$ nodes, where the edge labels are drawn from an alphabet of size $\sigma$,
the following trivial bounds are known:
We can reserve for each node $\sigma$ entries in $\barr$. If we rank all nodes from $1$ to $M$, then we store the ranks of the children of node $i$ in the range $\barr[1+ (i-1)\sigma  .. i \sigma-1]$ and set $\barr[i] = i$.
We further set $\carr[j] = i$ if $i$ has a child connected with an edge with label $j-i+1 \in [1..\sigma]$.
So both arrays have the length $M \sigma$,
where the fraction of empty space in both arrays grows with the alphabet size.
The best case we can hope for is to shrink the length down to $M$, which can be attained easily if the alphabet is unary.
However, the lower bound of $\Omega(M)$ does not seem to be achievable in the general case.
Apart from these two border cases, we are unaware of any studies on theoretical bounds for the space complexity of the double-array.
The lack of computational optimization makes it difficult to use double-arrays in a wider range of applications.
In fact, applications of double-arrays are still limited to natural language processing and information retrieval, although those of tries go beyond that scope.

On the computational side, an unpublished manuscript by Even et al.~\cite{even1997} shows that the problem is NP complete,
even when restricting the smallest size to $M+2$,
but the reduction seems to require an alphabet size of $\Theta(M)$\footnote{Although cited by many references, we were not able to obtain the manuscript.
    However, we will reconstruct a likely proof (cf.~\cref{thmSMC}\JO[~in the appendix]{}) following a note in the textbook~\cite[Chapter A4.2, Problem SR13]{DBLP:books/fm/GareyJ79} mentioning a reduction from Graph 3-Colorability.
}.
It thus remained an open question to us whether the problem is still NP-hard for a smaller alphabet size.
In this paper, we answer this question partially by proofing NP-hardness for alphabet sizes of $\Omega(\lg n)$,
shrinking the range to alphabet sizes between $3$ and $o(\lg n)$, for which the question is left open.

\block{Related Work}
There are several practical studies to compute smaller double-arrays for large alphabets.
Examples are Chinese strings \cite{liu2011compression}, Japanese strings \cite{kanda2023engineering}, and $N$-gram word sequences \cite{norimatsu2016fast}.
Their common solution is code mapping based on character frequencies and modifies the alphabet so that 
more frequent symbols have smaller ranks.
However, those methods assume that the alphabet has a strong bias in frequency of character occurrence, and we have no guarantees about their space efficiency.
In fact, we are unaware of any nontrivial theoretical guarantees.
Most of other studies on the double-array construction address compression aspects 
such as \cite{yata07compact} proposing a DAG-like compression of common suffixes or~\cite{kanda16compression} using linear-regression techniques, 
or address dynamic settings (e.g., \cite{kanda2017rearrangement}), which go beyond the scope of this paper.

\section{Preliminaries}

\block{Tries}
We can regard a trie as a directed labeled tree $G=(V, E)$, where $V = \{1,2,\dots,M\}$ is the set of nodes and
$E \subseteq V\times \Sigma \times V$ is the set of edges for the alphabet $\Sigma=\{1,2,\dots,\sigma\}$.
We write $(i,c,j) \in E$ for an edge from node $i \in V$ to an edge $j \in V$ with label $c \in \Sigma$.
By definition, there is exactly one node $u$ that has no in-coming edges $(\cdot,\cdot,u)$, which is the root of the trie.
Given the trie $G$ stores a string $S$, then there is a path from the root to a node $v$ such that reading the edge labels along this path gives $S$.

\block{Double-Arrays}
The double-array represents $E$ using two one-dimensional arrays, called $\barr$ and $\carr$, of the same length $N$.
The double-array arranges nodes $i \in V$ onto the arrays and assigns unique IDs $\node{i} \in \{1,2,\dots,N\}$ to them.
If and only if there exists an edge $(i,c,j) \in E$, the double-array satisfies the following conditions:

\begin{enumerate}
    \item $\barr[\node{i}]+c=\node{j}$, modelling that the $c$-th child of $i$ is $j$, and
    \item $\carr[\node{j}]=\node{i}$, modelling that the parent of $j$ is $i$.
\end{enumerate}

Because of these conditions, $\barr$ and $\carr$ can include vacant elements, i.e., $N$ can be larger than $M$.
We fix $N=\max\{\node{i} \mid i \in V\}$ unless otherwise noted since array elements with indexes greater than $\max \{\node{i} \mid i \in V\}$ are never used and redundant.
We assume vacant elements store $\None$.
As long as the above conditions are satisfied, the double-array allows any arrangement of nodes;
therefore, there are multiple arrangements, resulting in different values of $N$.

\block{Problem Formulation}
In the remaining of this paper, we study pattern matching of strings over the alphabet $[1,n]$ involving a wildcard symbol $\circ$.
The wildcard symbol $\circ$ matches with any of the characters in $[1,n]$.
We stipulate that two strings with wildcards $S_1$ and $S_2$ can match only if they have the same length.
We also say that a string $S_1$ occurs in $S_2$ if $S_2$ has a substring that matches with $S_1$ (involving wildcard symbols).
By doing so, we can now express our problem in terms of matching with wildcards:

\begin{problem}[Space Optimal Double-Array (\SODA{})]\label{problem:soda}
Given a set $\mathcal{S} = \{ S_1,\ldots, S_n\}$ with $S_i \in \{ i,\circ \}^{\sigma}$ for all $i\in [1,n]$, find the shortest string $S$ such that for all $i\in[1,n]$, $S_i$ has a match in $S$.
\end{problem}
By definition, even if $\mathcal{S} = \{ \circ^n \}$, the solution of Problem~\ref{problem:soda} is of length $n$.
\begin{example}
    If none of the strings contains the wildcard symbol $\circ$, then $s = s_1 s_2 \cdots s_n$ is the solution of this \SODA{} instance. Otherwise, the solution can be shorter.
    Given $\mathcal{S} = \{S_1,S_2\}$ such that $S_1$ (resp.\ $S_2$) has a wildcard symbol on every even (resp.\ odd) position, then $S[1..\sigma] = 2 1 2 1 \cdots$ of length $\sigma$ matches both $S_1$ and $S_2$.
\end{example}

For tiny alphabet sizes $\sigma \le 3$ it is easy to solve \SODA{}, but we are unaware of any solutions for larger (even still constant) alphabet sizes.
\begin{lemma}
    For the case $\sigma = 2, 3$, \SODA{} can be solved in polynomial time.
\end{lemma}
\begin{proof}
    For $\sigma = 2$, the input strings cannot model holes, meaning that $\circ$ must be either a prefix or a suffix of an input string. We therefore can start building the output string $Q$ greedily from left to right, first taking input strings of the form $S_i = i \circ$, then of the form $S_i = \circ i$, and finally of the form $S_i = i i$, where we always fill up the available wildcard.
    For $\sigma = 3$, we can have a wildcard in the middle. Only such strings can be combined, giving a new string without wildcards. The other cases are analogous to the case for $\sigma = 2$.
\end{proof}

We can reduce the problem of finding the minimal length for $\carr$ to \SODA{}.
Given $\barr[\node{i}] = x$,
we produce $S_{\node{i}}$ by copying
$\carr[x+1..x+\sigma]$ to $S_{\node{i}}$ and replacing all symbols different to $\node{i}$ with $\circ$.

\newcommand*{\RDHPProblem}{\textsc{RDHP}}
\newcommand*{\SCSProblem}{\textsc{SCS}}

\section{Hardness of \SODA}

In what follows, we want to study the computational complexity of \SODA{}.
We call $k$-\SODA{} the decision version of the problem \SODA{}.
For a given integer $k$,
the problem $k$-\SODA{} asks whether the length of the shortest string $S$ is at most $\sigma + k$.
This problem is a special case of the \textsc{Shortest Common Superstring}~(\SCSProblem) problem with wildcards.
Unfortunately, we have the following result for $k$-\SODA{}.

\begin{theorem}\label{thmKSODA}
    $k$-\SODA{} is NP-complete.
\end{theorem}

The problem has appeared in the literature as the \textsc{Sparse Matrix Compression} or \textsc{Compressed Transition Matrix}~\cite[Sect.~4.4.1.3]{martin2010scientific} problem:
\begin{problem}[{\cite[Chapter A4.2, Problem SR13]{DBLP:books/fm/GareyJ79}}]\label{problemSR13}
Given an $m\times n$ matrix $A$ with entries $a_{ij}\in \{ 0, 1 \}$, $1 \leq i \leq m$, $1 \leq j \leq n$, and a positive integer $k\leq mn$,
determine whether there exists a sequence $(b_1, \ldots, b_{n+k})$ of integers $b_i$,
each satisfying $0 \leq b_i \leq m$, and a function $s:\{1,\ldots, m\} \rightarrow \{1,\ldots, k\}$ such that, for $1\leq i\leq m$ and $1 \leq j \leq n$, the entry $a_{ij} = 1$ if and only if $b_{s(i)+j-1}=i$.
\end{problem}
This problem is NP-hard, even if $k=3$ by a reduction from three-coloring (see \cref{thmSMC}\JO[~in the appendix]{}).
\J{\begin{theorem}\label{thmSMC}
        \textsc{Sparse Matrix Compression} is NP-hard, even when the maximum shift is upper bounded by $k = 3$.
    \end{theorem}
    \begin{proof}
    Given a graph $(V,E)$ with $V := \{v_1, \ldots, v_n\}$ and edges $E$,
    we create an $n \times 3n$ matrix $M$ in the spirit of an adjacency matrix, where each adjacency matrix cell is expanded by three columns for assigning three colors.
    The details are as follows.
    We initialize $M$ by zero.
    Each row encodes the color of a vertex.
    We partition the $3n$ columns of the $i$-th row in blocks of size three,
    and write in the first entry of the $j$-th block a '1' if $(v_i,v_j) \in E$ or $i=j$
    (so each node has an artificially created self-loop).
    The semantics is that, if we think of the three entries of a block as an assignment for the colors red, blue, and green, then we initially color all vertices as red.
    By shifting the $i$-th row by one or two, we change the color of $v_i$ to green or blue, respectively.
    The color assignment is well-defined since such a shift shifts the '1's in all blocks of $v_i$ at the same time.
    Now, if we can shift all rows by an offset within $\{0,1,2\}$ such that we obtain the above sequence $(b_1,\ldots, b_{n+3})$, then the graph is three color-able.
\end{proof}
}
Although proofing Theorem~\ref{thmKSODA},
such a reduction makes no assumption on the range, in which we expect non-zero entries in a row. This is untypical to our use case of the double-array, where we only expect non-zero entries in a range of $\sigma$ adjacent entries.
We are unaware of restricted NP-hard problems based on three coloring, where the edges can be arranged that the region of non-zero entries can be bounded.

For our proof, we therefore follow a different approach, with which we obtain NP-hardness for $\sigma \in \Omega(\lg n)$.
To this end, we modify the proof of NP-completeness for the \SCSProblem{} problem~\cite{gallant1980finding}
which gives a reduction from the \textsc{Restricted Directed Hamiltonian Path} problem defined below.

\begin{problem}[\textsc{Restricted Directed Hamiltonian Path}~(\RDHPProblem{})~\cite{gallant1980finding}]\label{problem:rdhp}
Given a directed graph $G = (V, E)$, a designated start node $s\in V$
and a designated end node $t\in V$ such that
$s$ does not have incoming edges,
$t$ does not have out-going edges, and all nodes except $t$ have out-degree greater than $1$,
answer whether there exists a path from $s$ to $t$ that goes through each node exactly once.
\end{problem}

\begin{lemma}[Lemma~1 of \cite{gallant1980finding}]
    \RDHPProblem{} is NP-complete.
\end{lemma}

\subsection{A reduction to \SCSProblem{}}
We first revisit the proof of the following Theorem~\ref{thm:superstring} by Gallant et al.~\cite{gallant1980finding} who reduced \RDHPProblem{} to \SCSProblem{}.
We will not use the restriction of primitiveness, and will only use the length restriction $H=3$.
Later, we will show how to modify the reduction into an instance of \SODA{}.

\begin{theorem}[Theorem~1 of Gallant et al.~\cite{gallant1980finding}]\label{thm:superstring}
    \SCSProblem{} is NP-complete.
    Even if all input strings have the length $H$, for any integer $H \geq 3$,
    and that each character occurs at most once in $S$,
    the problem remains NP-complete.
\end{theorem}

Given $G = (V, E)$, let $|V| = n, |E|=m$, and $V = \{v_1 = s, \ldots, v_{n} = t\}$,
consider the
alphabet
$\Sigma =
    \{ \ddagger, \#, \$ \}
    \cup (V\setminus \{ v_1 \})
    \cup W
$, where $W = \{ w_1, \ldots, w_{n-1} \}$ are additional letters.
The conceptional idea is that $\ddagger$ and $\$$ will appear only once in the input set of strings and delimit our superstring at the start and at the end. These two special symbols, together with $\#$ will appear only in the strings $C_i$ that are used for modelling the connection of nodes.

\newcommand*{\Children}{\textup{children}}
For each node $v_i$ with $i \in \{1,\ldots, n-1\}$,
let $\delta(i)$ denote the out-degree of node~$v_i$,
and let $c(i,j)$ denote the node number corresponding to the $((j\bmod\delta(i))+1)$-th child of $v_i$ in some (arbitrary) order.
Given $\Children(i) := \{ v_{c(i,1)}, \ldots v_{c(i,\delta(i))} \}$ denotes the set of children of $v_i$,
we define the strings
$A_{i,j} := w_i v_{c(i,j)} w_i \in \Sigma^3$ and
$B_{i,j} := v_{c(i,j)} w_i v_{c(i,{j+ 1})} \in \Sigma^3$
for each $v_{c(i,j)} \in \Children(i)$.
Let $\mathcal{A}_i := \{A_{i,j} \mid v_{c(i,j)} \in R_{i}\} \cup \{B_{i,j} \mid v_{c(i,j)} \in R_{i}\}$ be the set containing these strings.
Finally, let
$C_1 = \ddagger \# w_1$,
$C_i = v_i \# w_i$ for $i\in\{ 2,\ldots, n-1\}$,
$C_n = v_n\#\$ \in \Sigma^3$, and $\mathcal{C} := \{C_1, \ldots, C_n\}$.
Our claim is that $G$ has a directed Hamiltonian path if and only if the set of strings
\(
\mathcal{T} = \mathcal{C} \cup \bigcup_{i=1}^{n-1} \mathcal{A}_i
\subset \Sigma^3
\)
has a superstring of length $2m+3n$.

For the proof we define, for two strings $X$ and $Y$, $X \bowtie Y$ to be the concatenation of $X$ and $Y[\ell+1..]$,
where $Y[1..\ell]$ is the longest prefix of $Y$ being a suffix of $X$.
Then $A_{i,j} \bowtie B_{i,j} = w_i v_{c(i,j)} w_i v_{c(i,{j+ 1})}$ and $B_{i,j} \bowtie A_{i,j+1} = v_{c(i,j)} w_i v_{c(i,{j+ 1})} w_i$.

Now assume there is a Hamiltonian path $v_1 = v_{i_1}, v_{i_2},\ldots, v_{i_{n-1}}, v_{i_n} = v_n$.
For each edge $(v_{i_k},v_{i_{k+1}})$,
we connect the strings in $\mathcal{A}_{i_k}$ with $\bowtie$ in a particular order to obtain the following string:
\begin{align}\label{eqSik}
    \begin{split}
        S_{i_k} &= A_{i_k,j} \bowtie B_{i_k,j} \bowtie A_{i_k,j+1} \ldots \bowtie A_{i_k,j+\delta(i_k)-1} \bowtie B_{i_k,j+\delta(i_k)-1}
        \\
        &= w_{i_k} v_{c(i_k,j)} w_{i_k} v_{c(i_k,j+1)} w_{i_k} \cdots v_{c(i_k,j+\delta(i_k)-1)} w_{i_k} v_{c(i_k,j+\delta(i_k))},
    \end{split}
\end{align}
where $j$ is such that $c(i_k,j) = i_{k+1}$.
Notice that from the definition of $c(\cdot)$, $v_{c(i_k,j+\delta(i_k))}= v_{c(i_k,j)}$,
and $|S_{i_k}| =2\delta(i_k)+2$.
Next, in the order of the Hamiltonian path,
we connect $C_{i_k}$ and $C_{i_{k+1}}$ with $S_{i_k}$ in between (see Fig.~\ref{fig:structure}),
starting with $\ddagger\#w_1 = C_1$ and ending with $v_n\#\$ = C_n$,
i.e., in the order
\begin{align*}
    S : & = C_1\bowtie S_1\bowtie C_{i_2}\bowtie S_{i_2}\bowtie C_{i_3}\bowtie S_{i_3}\bowtie \cdots\bowtie S_{i_{n-1}}\bowtie C_n                               \\
        & =\ddagger \# w_1\bowtie S_1\bowtie v_{i_2}\# w_{i_2}\bowtie S_{i_2}\bowtie \ldots \bowtie v_{i_{n-1}}\# w_{i_{n-1}}\bowtie S_{i_{n-1}}\bowtie v_n\#\$,
\end{align*}
where each string is connected to the previous one with an overlap of 1 symbol.
$S$ is a superstring of $\mathcal{T}$.

On the one hand, the total length of $S$ is
$2m+2(n-1) + 2\cdot 2 + n-2 = 2m+3n$,
which we obtain by the facts that
\begin{enumerate}
    \item $\sum_{i=1}^{n-1}|S_i| = 2m+2(n-1)$ with $\sum_{i=1}^{n-1} \delta(i) = m$,
    \item the two strings $C_1, C_n$ further contribute a length of $2$ each (for $\ddagger\#$ and $\#\$$ occurring neither in $S_{1}$ nor in $S_{i_{n-1}}$),
          and
    \item the strings $C_i$ ($i\in\{2,\ldots,n-1\}$) further contribute a length of $1$ each (for the $\#$'s that do not occur in the $S_i$'s).
\end{enumerate}

On the other hand, suppose we are given a superstring $Q$ of $\mathcal{T}$ of length $2m+3n$.
The total number of strings in $\mathcal{T}$ is $2m+n$.
\begin{itemize}
    \item Since all strings in $\mathcal{T}$ are distinct and not a substring of another, a lower bound on the shortest superstring is $3 + (2m+n-1) = 2m+n+2$,
          where all strings are connected in some order by overlapping two symbols with its previous string.
    \item Let us consider $C_i$ with $i\in\{2,\ldots,n-1\}$ as a substring appearing in $Q$.
          Because $\#$ only occurs in the strings of $\mathcal{C}$, which have all no overlaps, the largest overlap is gained by the combinations $C_i \bowtie A_{i,j}$ for some $j$ or
          $B_{i,j} \bowtie C_{c(i,j+1)}$, which are overlaps of lengths one.
\item Similarly, the two strings $C_1$ and $C_n$ cannot overlap with any other string on one side, and at most $1$ on the other side with $A_{1,j}$ for some $j$ or $B_{i,k}$ for some $k$ with $c(i,k+1) = n$.
\end{itemize}
Thus, the lower bound can further be improved to $2m+n+ 2(n-2)+2 + 2 = 2m+3n$.
This also implies that $Q$ must start with $\ddagger\#w_1= C_1$ and end with $v_n\#\$= C_n$,
since otherwise, $C_1$, or the string following $C_n$ could not have exploited any overlap with its previous string,
and the resulting string would have been longer.

It is left to show that we can read a Hamiltonian path from $Q$, which starts with $v_1$ and ends with $v_n$.
To this end, consider any substring of $Q$ between consecutive occurrences of $\#$
corresponding to the occurrence of $C_{i_k}$ and $C_{i_{k+1}}$,
where the first symbol is $w_{i_k}$ and the last symbol must be $v_{i_{k+1}}$.
Since $C_{i_k}$ must overlap by one with its following string, the following string must be $A_{i_k,j} \in \mathcal{A}_{i_k}$ for some $j$.
Furthermore, since strings in $\mathcal{A}_{i_k}$ can only overlap by two with a string in $\mathcal{A}_{i_k}$,
all and only elements of $\mathcal{A}_{i_k}$ must occur
starting with $A_{i_k,j} = w_{i_k}v_{c(i_k,j)}w_{i_k}$ for some $j\in\{1,\ldots,\delta(i_k)\}$, and ending with
$B_{i_k,j+\delta(i_k)-1} = v_{c(i_k,j+\delta(i_k)-1)} w_{i_k} v_{c(i_k,j+\delta(i_k))} = v_{c(i_k,j+\delta(i_k)-1)} w_{i_k} v_{c(i_k,j)}$.
From here on, we have already used up all strings starting with $w_{i_k} v_{c(i_k,j)}$, and thus the maximum overlap to gain is of length 1.
However, at this point we can append $\bowtie C_{i_{k+1}}$, letting it overlap with $B_{i_k,j+\delta(i_k)-1}$ for maximal overlap.
Hence, $v_{i_{k+1}} := v_{c(i_k,j)}$, and we obtain
a Hamiltonian path $v_1 = v_{i_1},\ldots, v_{i_n}=v_n$ by iterating these steps for all $k$.

\begin{figure}[htbp]
    \centerline{\includegraphics[width=0.9\textwidth,page=1]{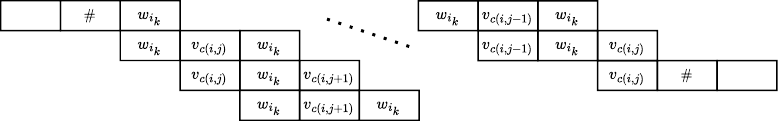}
    }
    \caption{The structure of string $S_{i_k}$ in Eq.~\ref{eqSik} connected with $C_{i_k}$ to its left and $C_{i_{k+1}}$ to its right, where $c(i_k,j) = i_{k+1}$.
        The leftmost blank box can be $v_{i_k}$ or $\ddagger$ (if $i_k = 1$).
        The rightmost blank box can be $w_{c(i_k,j)}$ or $\$$ (if $i_k = n$).}\label{fig:structure}
\end{figure}

\subsection{A reduction to \SODA}
In what follows, we modify each of the strings obtained in the reduction to \SCSProblem{} to strings of length $\ell = O(\log n)$, which will be specified more precisely later, for \SODA{}.
We first consider assigning a distinct bit string of
length $O(\log n)$ to each of the symbols in $V\cup W$.
We further modify these bit strings
by applying the string morphism $\phi$ with $\phi(0) = 01$ and $\phi(1) = 10$, and prepending $\alpha := 01000011$ to the resulting bit string.
For $u\in V\cup W$, we denote the resulting bit string as $b(u)$,
whose length is denoted as $\ell'$.
The idea is that all bit strings $b(u)$ start with $\alpha$, while the suffixes $b(u)[\alpha+1..]$ have the same number of ones and zeros thanks to the definition of $\phi$.

Further, let $\overline{b(u)}$ denote the bit string where $0$ and $1$ are flipped.
With another function $g$, we will later map $0$ to the wildcard symbol $\circ$ and $1$ to a symbol unique to the respective input string.
If we extend $\bowtie$ for wildcard matching such that, for two strings $X$ with $Y$ with wildcard symbols $\circ$,
$X \bowtie Y$ denotes the shortest string $S$ such that $X$ and $Y$ match as a prefix and a suffix of $S$, respectively.
We say that two strings with wildcards $X$ and $Y$ of the same length are \emph{compatible} if $|X| = |X \bowtie Y|$.
By construction, $g(b(c))$ is only compatible with $g(\overline{b(c)})$ among all strings in $\{g(b(d)) : d \in V \cup W \}$, for every $c \in V \cup W$.

\newcommand*{\effL}[1][]{\ensuremath{f_{\textup{L}#1}}}
\newcommand*{\effM}[1][]{\ensuremath{f_{\textup{M}#1}}}
\newcommand*{\effR}[1][]{\ensuremath{f_{\textup{R}#1}}}

Remembering the definition of the set of input strings $\mathcal{T}$,
we assign each string of $\mathcal{T} \cup V \cup W$ a unique ID via a mapping $\iota$ from these strings to integers.
Now, let $xyz \in\mathcal{T}$, such that $x\in V\cup W\cup\{\ddagger\}$, $y\in V\cup W\cup\{\#\}$, and $z\in V\cup W\cup\{\$\}$.
We map $xyz$ to a string $f(xyz) = (\effL[,i](x)\cdot \effM[,i](y)\cdot \effR[,i](z)) \in \{i,\circ\}^{\ell}$
where $i = \iota(xyz)$,
and $\effL[,i], \effM[,i], \effR[,i]$ are defined as follows:

\(
\mycases{\effL[,i](c) = }{
    i^{3\ell'}                                  \text{if~} c =\ddagger, \text{~and} \\
        g_i(b(c))g_i(\overline{b(c)})\circ^{\ell'}  \mbox{~otherwise};\ 
    }
\)\JO{

}\(
    \mycases{\effM[,i](c) =}{
        i^{3\ell'}                                  \text{if~} c =\#, \text{~and} \\
        g_i(\overline{b(c)})\circ^{\ell'}g_i(b(c))  \mbox{~otherwise};\ 
    }
\)\JO{

}\(
    \mycases{\effR[,i](c) = }{
        i^{3\ell'}                                  \text{if~} c =\$, \text{~and} \\
        \circ^{\ell'}g_i(b(c))g_i(\overline{b(c)})  \mbox{~otherwise}. \
    }
    \)

\JO{\begin{eqnarray*}
    \effL[,i](c) &=& \begin{cases}
        i^{3\ell'}                                 & \text{if~} c =\ddagger \\
        g_i(b(c))g_i(\overline{b(c)})\circ^{\ell'} & \mbox{otherwise}
    \end{cases}\\
    \effM[,i](c) &=& \begin{cases}
        i^{3\ell'}                                 & \text{if~} c =\# \\
        g_i(\overline{b(c)})\circ^{\ell'}g_i(b(c)) & \mbox{otherwise}
    \end{cases}\\
    \effR[,i](c) &=& \begin{cases}
        i^{3\ell'}                                 & \text{if~} c =\$ \\
        \circ^{\ell'}g_i(b(c))g_i(\overline{b(c)}) & \mbox{otherwise.}
    \end{cases}
\end{eqnarray*}
}
$g_i$ is a morphism such that $g_i(0) = \circ$ and $g_i(1) = i$.
Therefore, each $f_\cdot(\cdot)$ gives a string of length $3\ell'$,
and $\ell = 9\ell'$.
If the ID $i$ is not of importance, we omit it in the subscripts, e.g., we write $\effL(x)$ instead of $\effL[,i](x)$.

For technical reasons, we further add some extra strings of length $3\ell'$ to the pool, which we call \emph{patch strings}, and consider the set
\begin{eqnarray*}
    \mathcal{T'} &=& \{ f(xyz) \mid xyz \in \mathcal{T}\}
    \cup\{ \effM[,\iota(w_k)](w_k), \effL[,\iota(w_k)](w_k) \mid k\in \{1,\ldots,n-1\}\}\\
    &&\cup\{ \effM[,\iota(v_k)](v_k), \effR[,\iota(v_k)](v_k) \mid k\in \{2,\ldots,n\}\}.
\end{eqnarray*}
We will show that $\mathcal{T'}$
has a superstring of length $(2m+3n)3\ell'$ if and only if $G$ has a Hamiltonian path from $s$ to $t$.
We have only to show that the given strings in $\mathcal{T}'$
can be made to overlap as the strings of $\mathcal{T}$,
and must overlap analogously to a shortest superstring containing all occurrences of strings in $\mathcal{T}'$.
To this end we use the following facts:
\begin{enumerate}
    \item The string $g_i(b(c)) g_i(\overline{b(c)}$ has $\ell'$ wildcard symbols. \label{itFactEllWildcards}
    \item The strings $\effL[,i](c), \effM[,j](c),$ and $\effR[,k](c)$ are compatible for $c \in U \cup V$ and any IDs $i,j,$ and $k$. \label{itFactEffCompat}
          The string $\effL[,i](x) \bowtie \effM[,j](x) \bowtie \effR[,k](x)$ has length $3\ell'$ and contains no wildcards.
    \item While encodings of any $c \in U \cap V \setminus \{V_1\}$ appear in any combination of $\effL(c), \effM(c),$ and $\effR(c)$,
          the special symbols $\ddagger,\#,$ and $\$$ occur only as substrings $\effL(\ddagger), \effM(\#),$ and $\effR(\$)$ of the strings in $\mathcal{T}$.
\end{enumerate}

On the one hand, suppose we are given a Hamiltonian path $v_{i_1}, \ldots v_{i_n}$ of $G$.
Using the construction of the string $S$ in the reduction to \SCSProblem{} with Fact~\ref{itFactEffCompat},
we can construct a superstring of the claimed length that contains the occurrences of all strings in $\mathcal{T}'$.
We use the extra strings of length $3\ell'$ in $\mathcal{T'}$
to ``patch'' locations where only two strings from $\{f(xyz)\mid xyz\in\mathcal{T}\}$ overlap,
with the intent to enforce the elimination of all wildcards in the shortest superstring.
In Figure~\ref{fig:structure},
we can apply the extra strings $f_M(w_{i_k})$ and $f_R(v_{c(i,j)})$ for the leftmost occurrences of $w_{i_k}$ and $v_{c(i,j)}$,
and $f_L(w_{i_k})$ and $f_M(v_{c(i,j)})$ for the rightmost occurrences of $w_{i_k}$ and $v_{c(i,j)}$, cf. Fig.~\ref{fig:structure_enc}.

\begin{figure}
    \centering
    \includegraphics{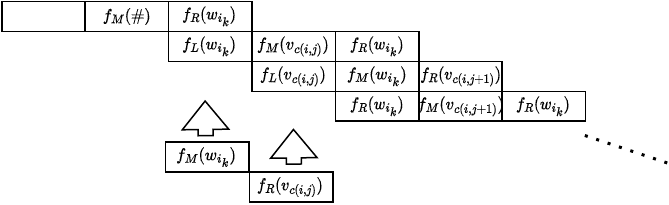}
    \caption{Constructing a shortest superstring for the input string set $\mathcal{T}'$ with the same building blocks as in Figure~\ref{fig:structure}.
        The patch strings fill up all wildcard symbols.}
    \label{fig:structure_enc}
\end{figure}

On the other hand, suppose a superstring $Q$ of $\mathcal{T}'$ of length $(2m+3n)3\ell'$ exists.
The number of non-wildcard symbols in $f_L(c)$, $f_M(c)$, or $f_R(c)$
for $c\in V\cup W$ is $\ell'$ each,
so the total number of non-wildcard symbols in $\mathcal{T'}$ is
$3\ell'$ for each of the $2m$ strings in $\bigcup_{i=1}^{n-1}\{ f(xyz)\mid xyz\in \mathcal{A}_{i} \}$,
and $5\ell'$ for each of the $n-2$ strings $\{ C_i \mid i\in\{ 2,\ldots, n-2\} \}$,
and $7\ell'$ for $C_1$ and $C_n$. Furthermore, the $4n-4$ extra strings contain $\ell'$ non-wildcard symbols each due to Fact~\ref{itFactEllWildcards},
for a grand total of $(2m+3n)3\ell'$.
Since non-wildcard symbols are distinct between different strings in $\mathcal{T'}$
and thus cannot overlap,
it follows that
$Q$ cannot contain any wildcard characters and that all wildcard characters must be ``filled'' when connecting the strings in $\mathcal{T}'$.

While in general, strings in $\mathcal{T}'$ can overlap in a way that does not respect the symbol boundaries in their underlying string in $\mathcal{T}$,
we claim that for any shortest superstring $Q$ of length $(2m+3n)3\ell'$,
the strings in $\mathcal{T}'$ must be connected in a way that they are {\em block aligned},
meaning that the length of the overlap with any other string in $\mathcal{T'}$ must be a multiple of $3\ell'$.

\begin{figure}[htbp]
    \begin{minipage}{0.3\linewidth}
        \includegraphics[width=\textwidth]{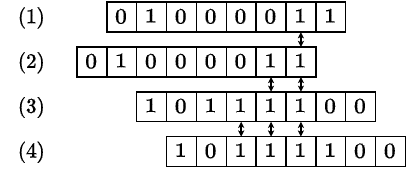}
    \end{minipage}
    \begin{minipage}{0.7\linewidth}
    \caption{Non-sub-block-aligned overlapping of $\alpha$ and $\overline{\alpha}$ to fill the leftmost wildcard. The figure assumes that there are no wildcards preceding (2).
        The wildcard (corresponding to $0$) in (2) may overlap with previously filled blocks. Since the $1$s are mapped to IDs distinct to each string, they must not overlap.
        For example, the overlap of (2) and (1) tries to fill the 2nd $0$ of (2) with another $\alpha$ (1), but is interfered with the last $1$ of (2) and the second to last $1$ of (2).
        Similarly, the overlap of (2) and (3) tries also to fill the 2nd $0$ of (2) with $\overline{\alpha}$ (3). The overlap of (3) and (4) tries to fill the first $0$ of (3) $\overline{\alpha}$ with $\overline{\alpha}$. Thus, the only way to fill the first $0$ in $\alpha$ or the first $0$ in $\overline{\alpha}$ is to align $\alpha$ and $\overline{\alpha}$.
    }\label{fig:alpha_overlap}
    \end{minipage}
\end{figure}

\begin{lemma}\label{lemma:block_aligned}
	If there exists a string $Q$ of length $(2m+3n)3\ell'$ such that all strings of $\mathcal{T}'$ have a match in $Q$,
then every $S \in \mathcal{T'}$ must have a unique {\em block-aligned} occurrence in $Q$, i.e.,
    the ending position of $S$ in $Q$ must be a multiple of $3\ell'$.
\end{lemma}
\begin{proof}
From the above arguments, if such a string $Q$ exists, then it cannot contain any wildcards.
    Consider building $Q$ from left to right, starting from a wildcard string of length $(2m+3n)3\ell'$, each time connecting a string in $\mathcal{T'}$
    to fill the leftmost remaining wildcard.
    Suppose the leftmost remaining wildcard is at position $p$ and consider what strings in $\mathcal{T}'$ can come next in order to determine the symbols for $Q[p..]$.

    \begin{figure}[h]
        \begin{minipage}{0.72\linewidth}
        \centerline{\includegraphics*[page=5,clip=true,width=\linewidth]{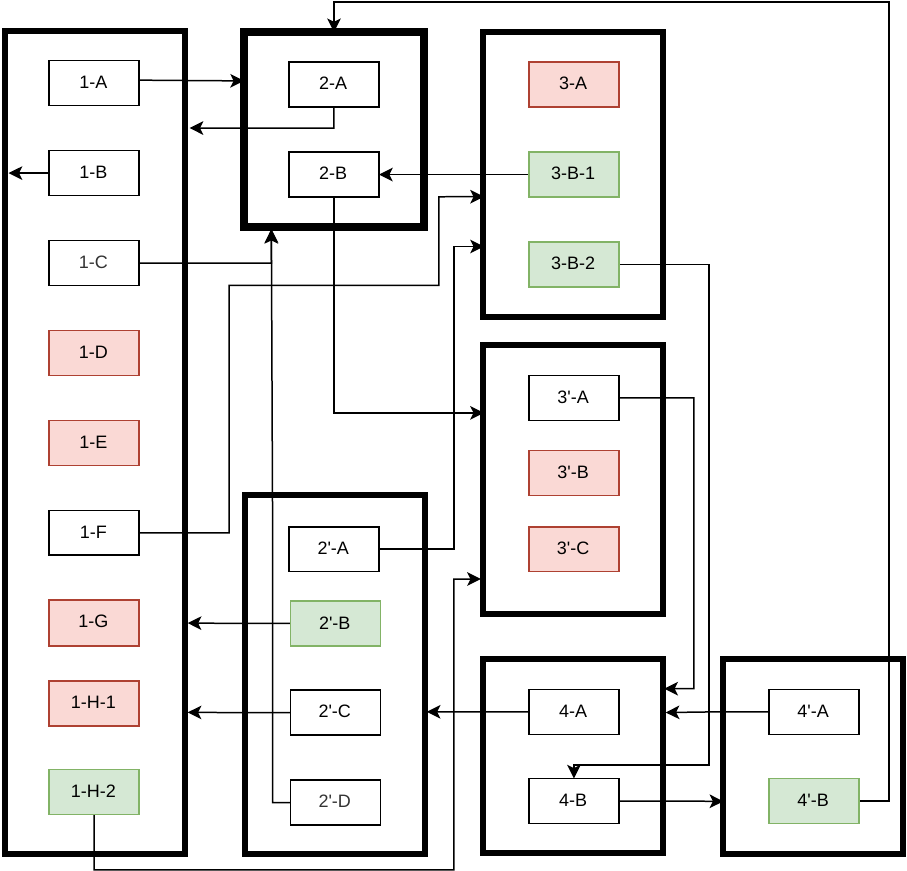}
        }
        \end{minipage}
        \begin{minipage}{0.27\linewidth}
        \caption{Overview of all cases considered in the proof of Lemma~\ref{lemma:block_aligned}. Subcases of Cases 1,2,2',3,3',4, and~4' are grouped in bold frames.
            A red subcase means that it is not possible to continue.
            Any other case has an outgoing edge to either a specific subcase, or to a case $X$ meaning that we can continue with any of the subcases of $X$.
            Nevertheless, despite that green subcases have an outgoing edge, they will not lead to the optimal solution.
            They are called \emph{prohibitive}, and subject in the proof of Lemma~\ref{lemma:block_aligned}.
            All other cases join the input strings of $\mathcal{T'}$ in a block-aligned manner.
        }\label{figure:transitionGraph}
        \end{minipage}
    \end{figure}

    We first argue that the strings must be \emph{sub-block aligned} in $Q$,
    i.e.,
    the ending position in $Q$, or equivalently, the overlap with any previous string in $\mathcal{T}'$ must be a multiple of $\ell'$.
    Assume that all strings previously connected are sub-block aligned, and consider filling the leftmost wildcard remaining in $Q[p..]$.
    Due to the definition of $f_L,f_M,f_R$, only compatible sub-blocks can overlap and a maximal overlap of (non-all wildcard) compatible sub-blocks will always result in a filled sub-block, cf.\ Fact~\ref{itFactEffCompat}.
    Thus, remaining non-filled sub-blocks must either be all wildcards, $c$, or $\overline{c}$ for some $c$.
    All non-all wildcard sub-blocks of strings in $\mathcal{T}'$
    and remaining non-filled/non-all wildcard sub-blocks start
    with either $\alpha$ or $\overline{\alpha}$.
Fig.~\ref{fig:alpha_overlap} shows all the possible cases
    when trying to fill the left-most wildcard of an $\alpha$ (or $\overline{\alpha}$)
    when the strings are not sub-block aligned.
    In each case, the occurrences are not compatible and thus can only be sub-block aligned.
    Thus, by induction, all occurrences must be sub-block aligned.

    Next, we show that the strings must be block-aligned in $Q$.
    Since the strings connected must be sub-block aligned,
    there are a total of $4$ main cases divided according to the shape of the prefix of $Q[p..]$.
    The transition between the cases is depicted in Fig.~\ref{figure:transitionGraph}, while the cases themselves
    are depicted in Fig.~\ref{figBlockAlignedCompact} (Figs.~\ref{figure:block_aligned_rem0} to~\ref{figure:block_aligned_rem21}\JO[~in the appendix]{} are larger versions).
    Case 1 is when $Q[p..]$ consists only of wildcard strings, which is the initial and final state when building $Q$ successively.
We thus can consider the building of $Q$ as a transition of cases starting from Case 1 and ending at Case 1.
    During the course of studying all cases, we observe that some cases have transitions but do not block-align.
    We call these cases \emph{prohibitive}, and will subsequently show that prohibitive cases cannot lead to a shortest superstring.

    \newcommand*{\Filled}{\XBox}
    \newcommand*{\Empty}{\Square}

    For the detailed description, 
for any $c\in V\cup W$, we denote $g_i(b(c))$ as $c$ and $g_i(\overline{b(c)})$ as $\overline{c}$.
Further,
    we write \Filled{} for any string of length $\ell'$ containing no wildcards,
    and $\Empty := \circ^{\ell'}$ for the opposite case.
    To discriminate the cases, $Q[p..]$ is considered as
    \begin{description}
        \item[Case 2]:  $\effR(w_i) \circ \cdots$,
        \item[Case 2']:  $\effR(v_{c(i,j})) \circ \cdots$,
        \item[Case 3]:  $\effM(w_i) \effR(v_{c(i,j)}) \circ \cdots$, matching with a sub-block suffix of $B_{i,j}$,
        \item[Case 3']:  $\effM(v_{c(i,j)}) \effR(w_i) \circ \cdots$, matching with a sub-block suffix of $A_{i,j}$,
        \item[Case 4]:  $\overline{w_i} w_i \Filled \effR(v_{c(i,j}) \circ \cdots$, matching with a sub-block suffix of $B_{i,j}$,
        \item[Case 4']:  $\overline{v_{c(i,j)}} v_{c(i,j)} \Filled \effR(w_i) \circ \cdots$, matching with a sub-block suffix of $A_{i,j}$,
    \end{description}
    for some $i$ and $j$.

    Starting with Case 1, which also treats $Q$ in its initial state storing only wildcards, we are free to choose any of the strings in $\mathcal{T'}$.
    The first case (1-A) is special in that it involve the insertion of only one string of $\mathcal{T'}$.
    \begin{description}
        \item[1-A]: Start with $f(C_1)$. The rightmost sub-block with wildcards in $Q[p..] \bowtie f(C_1)$ is $\effR(w_1)$, a transition to Case 2.
    \end{description}
    In all other cases, regardless of the inserted string of $\mathcal{T}$, we can fill up its leftmost wildcards only with one of the two patch strings $\effM(v_i)$  or $\effM(w_i)$ for some $i$.
    We start with the cases for $\effM(v_i)$:
    \begin{description}
        \item[1-B]: Start with $\effM(v_i) \bowtie f(C_n) = \Filled \overline{v_n} v_n \# \$$. To fill the leftmost wildcards,
        we can only apply the patch string $\effR(v_n) \in \mathcal{T'}$ such that $Q[1..p+\ell-1] = \Filled^3 \# \$$ is completely filled.
        We stay in Case 1.
        \item[1-C]: Start with $\effM(v_i) \bowtie f(C_{i}) = \Filled \overline{v_{i}} v_i  \effM(\#) \effR(w_{i})$ for $i \in [2..n-1]$.
        To fill the wildcards of $\overline{v_{i}} v_i$, we must patch it similar to Case 1-B with $\effR(v_i) \in \mathcal{T'}$.
        This time, the right part of $\effM(v_i) \bowtie f(C_{i}) \bowtie \effR(v_i) = \Filled^3 \# \effR(w_{i})$ still contains wildcards, a transition to Case~2.
\item[1-D]: Start with the patch strings $\effR(v_i)$ and $\effM(v_i)$ for $p \ge \ell'$.
        We have $\effR(v_i) \bowtie \effM(v_i) = \Filled \overline{v_i} v_i$.
        Since $\effL(v_i)$ is not part of $\mathcal{T'}$, the only choice left is to fill it up with $C_i$ ($i \in [2..n]$) or $B_{i',j}$ with $i = c(i',j)$.
        In these cases, $Q[p..] \bowtie B_{i',j}$ starts with $\Filled^3 \Empty \overline{w_{i'}} \Empty$ and $Q[p..] \bowtie C_i$ starts with $\Filled^3 \Empty \overline{w_{i'}} \#^\ell$.
        While we can still apply $\effM(w_{i'})$ for $Q[p..] \bowtie B_{i',j}$, in none of the cases we can completely fill the wildcards with any of the remaining strings in $\mathcal{T'}$, so Case 1-D cannot lead to a solution.
        \item[Cases 1-E and 1-F]:
        Start with $\effM(v_{c(i,j)}) \bowtie f(B_{i,j}) = \Filled \overline{v_{c(i,j)}} v_{c(i,j)} \effM(w_i) \effR(v_{c(i,j)})$.
        Our aim is to patch the leftmost wildcards with some string containing $v_{c(i,j)} \overline{v_{c(i,j)}}$, for which we have two options:
        \begin{description}
            \item[1-F]: Applying the patch string $\effR(v_{c(i,j)})$ gives us $Q[p..] \bowtie \effM(v_{c(i,j)}) \bowtie f(B_{i,j}) \bowtie \effR(v_{c(i,j)}) = \Filled^3 \effM(w_i) \effR(v_{c(i,j)}) \circ \cdots$, which is in Case~3.
            \item[1-E]: Applying $B_{i',j'}$ for some $(i',j') \neq (i,j)$ but $c(i,j) = c(i',j')$,
            lets $Q[p..]$ start with $\Filled^3 \overline{w_i} \overline{w_{i'}} w_i w_{i'}$, which does not match with any of the remaining strings in $\mathcal{T'}$.
        \end{description}
    \end{description}

    We now analyze the remaining cases, where $\effM(w_i)$ must be applied:
    \begin{description}
        \item[1-G]: Start with the patch strings $\effM(w_i) \bowtie \effL(w_i)$, leaving a remaining suffix starting with $\overline{w_i} w_i$.
        This matches with $\effR(w_i)$, which is not a patch string, but a prefix of $f(A_{i,j})$.
        $\effM(w_i) \bowtie \effL(w_i) \bowtie f(A_{i,j})$ starts with $\Filled^3 \Empty \overline{v_{c(i,j)}} \Empty$.
        While we can still apply $\effM(v_{c(i,j)})$ here, we cannot completely fill the wildcards with any of the remaining strings in $\mathcal{T'}$, so Case 1-G cannot lead to a solution.
        \item[Cases 1-H]: Start with $\effM(w_i) \bowtie f(A_{i,j}) = \Filled \overline{w_i} w_i \effM(v_{c(i,j)}) \effR(w_i)$.
        Our goal is to fill the left part $\overline{w_i} w_i$.
        Symmetrically to Cases 1-E and 1-F we can apply $A_{i,j'}$ for some $j'$, or the patch string $\effL(w_i)$.
        \item[1-H-1]: $\effM(w_i) \bowtie f(A_{i,j}) \bowtie A_{i,j'}$ starts with $\Filled^3 \overline{v_{c(i,j)}} \overline{v_{c(i,j')}} v_{c(i,j)} v_{c(i,j')}$.
        Unfortunately, this substring does not match with any of the strings in $\mathcal{T'}$ because $c(i,j) \neq c(i,j')$.
        \item[1-H-2]: $Q[p..] \bowtie \effM(w_i) \bowtie f(A_{i,j}) \bowtie \effL(w_i) = \Filled^3 \overline{v_{c(i,j)}} \Empty v_{c(i,j)} \Empty w_i \overline{w_i} \circ \cdots$,
        a transition to Case 3'.
        This case is prohibitive because we did not block-align $\effL(w_i)$.
    \end{description}

    Case 2 considers $Q[p..] = \effR(w_i) \circ \cdots$.
    To fill the leftmost wildcards, we have to apply the patch string $\effM(w_i) \in \mathcal{T}$ and either
    \begin{description}
        \item[2-A]: the patch string $\effL(w_i) \in \mathcal{T}$, or
        \item[2-B]: $f(A_{i,j})$ for some $j$.
    \end{description}
    Application of Case 2-A gives $Q[p..] \bowtie \effR(w_i) \bowtie \effM(w_i) \bowtie \effL(w_i) = \Filled^3 \circ \cdots$, so we return to Case~1.
    Application of Case 2-B gives $Q[p..] \bowtie \effR(w_i) \bowtie f(A_{i,j}) = \Filled^3 \effM(v_{c(i,j)}) \effR(w_i) \circ \cdots$, which is the starting point of Case~3'.

    Case 2' considers $Q[p..] = \effR(v_{c(i,j)}) \circ \cdots$.
    Similar to Case~2, we have to apply $\effM(v_{c(i,j)}) \in \mathcal{T}$ to fill the leftmost wildcards.
    We thus start with $Q[p..] \bowtie \effM(v_{c(i,j)}) = \overline{v_{c(i,j)}} v_{c(i,j)} \Filled \circ \cdots$.
    Among the strings having $v_{c(i,j)} \overline{v_{c(i,j)}}$ as a substring, we are left with four options:
    \begin{description}
        \item[2'-A]: Apply $B_{i',j'}$ for any $i$ and $j$ with $c(i,j) = c(i',j')$.
        We have $\effR(v_{c(i,j)}) \bowtie B_{i',j'} = \Filled^3 \effM(w_{i'}) \effR(v_{c(i',j')})$, which is in Case 3.
\item[2'-B]: Apply $\effR(v_{c(i,j)})$ for $p \ge 3\ell'$. We have $Q[p..] \bowtie \effR(v_{c(i,j)}) = \Filled^3 \circ \cdots$.
        We return to Case~1.
        Because we applied $\effR(v_{c(i,j)})$ in a non-block-aligned manner, this case is prohibitive.
        \item[2'-C]: Apply $C_n$ in case that $c(i,j) = n$. With $\effR(v_{c(i,j)}) \bowtie f(C_n) = \Filled^3 \# \$$, this lets us again return to Case~1.
        \item[2'-D]: Apply $C_{c(i,j)}$. With  $\effR(v_{c(i,j)}) \bowtie f(C_{c(i,j)}) = \Filled^3 \# \effR(w_{c(i,j)})$, we move to Case~2.
    \end{description}

    Case 3 considers $Q[p..] = \effM(w_i) \effR(v_{c(i,j)}) \circ \cdots$.
    To fill the leftmost wildcards, we need to apply a string starting with $\effL(v_{c(i,j)})$,
    for which we can either use the patch string $\effL(v_{c(i,j)})$ itself, or $A_{i,j}$.
    As it turns out we need both strings, but the order of application matters:
    \begin{description}
        \item[3-A]: $Q[p..] \bowtie \effL(v_{c(i,j)}) = \Filled \overline{w_i} w_i \Empty \circ \cdots$.
        To fill up the leftmost wildcards we need a string containing $w_i \overline{w_i}$,
        for which $A_{i,j}$ only matches.
        However, $Q[p..] \bowtie \effL(v_{c(i,j)}) \bowtie A_{i,j} = \Filled^3 \Empty \Filled \overline{v_{c(i,j)}} \circ \cdots$.
        We cannot match any of the remaining strings in $\mathcal{T'}$ with this substring.
        \item[3-B]: $Q[p..] \bowtie f(A_{i,j}) = \Filled \overline{w_i} w_i \overline{v_{c(i,j)}} v_{c(i,j)} \Filled \effR(w_i) \circ \cdots$.
        Symmetric to Case~A, only the patch string $\effL(v_{c(i,j)})$ can fill up the leftmost wildcards, leading to
        $Q[p..] \bowtie f(A_{i,j}) \bowtie \effL(v_{c(i,j)}) = \Filled^{3} \overline{v_{c(i,j)}} v_{c(i,j)} \Filled \effR(w_i) \circ \cdots$.
        To fill the leftmost wildcards, we need a string containing $v_{c(i,j)} \overline{v_{c(i,j)}}$ as a substring, for which we have two options left:
        select the patch string $\effR(v_{c(i,j)})$ or $B_{i,j}$ (all other $B_{i',j'}$ do not match with the later $w_i$ occurrence).
        \begin{description}
            \item[3-B-1]: $Q[p..] \bowtie f(A_{i,j}) \bowtie \effL(v_{c(i,j)}) \bowtie \effR(v_{c(i,j)}) = \Filled^6 \effR(w_i) \circ \cdots$,
                a transition to Case~2.
                Since we have already used up $\effL(w_i)$, we cannot continue with Case~2-A, but have to continue with Case~2-B.
            \item[3-B-2]: $Q[p..] \bowtie f(A_{i,j}) \bowtie B_{i',j'} = \Filled^6 \overline{w_i} w_i \Filled \effR(v_{c(i,j}) \circ \cdots$,
                a transition to Case~4.
                Again, because we have already used $\effL(w_i)$, we need to continue with Case~4-B.
        \end{description}
        Because we applied $\effL(v_{c(i,j)})$ in a non-block aligned manner, both subcases are prohibitive.
    \end{description}

    Case 3' considers $Q[p..] = \effM(v_{c(i,j)}) \effR(w_i) \circ \cdots$.
    To fill the leftmost wildcards, we can apply $B_{i,j}$ or the patch string $\effR(v_{c(i,j)})$.
    \begin{description}
        \item[3'-A and 3'-B]: $Q[p..] \bowtie B_{i,j} = \Filled \overline{v_{c(i,j)}} v_{c(i,j)} \overline{w_i}  w_i \Filled \effR(v_{c(i,j}) \circ \cdots$.
        To fill the leftmost wildcards, we need a string having the substring $v_{c(i,j)} \overline{v_{c(i,j)}}$.
        From the strings of $\mathcal{T'}$  we can only choose the patch string $\effR(v_{c(i,j})$ or $B_{i',j'}$ with $c(i,j) = c(i',j')$ and $c(i,j+1) = c(i',j'+1)$.
        We start with the former.
        \item[3'-A]: $T[p..] \bowtie B_{i,j} \bowtie \effR(v_{c(i,j}) = \Filled^3 \overline{w_i}  w_i \Filled \effR(v_{c(i,j}) \circ \cdots$, a setting of Case~4.
        \item[3'-B]: $T[p..] \bowtie B_{i,j} \bowtie B_{i',j'}$ starts with $\Filled^3 \overline{w_i}  \overline{w_{i'}}$, but there is no string in $\mathcal{T'}$ matching $\overline{w_i}  \overline{w_{i'}}$.
\item[3'-C]:
            $T[p..] \bowtie \effR(v_{c(i,j)}) = \Filled \overline{v_{c(i,j)}} v_{c(i,j)} \effR(w_i) \circ \cdots$.
            Symmetrically to the previous case, we are in the need for a string containing $v_{c(i,j)} \overline{v_{c(i,j)}}$,
            where only $B_{i,j}$ is suitable (any other $B_{i',j'}$ does not match with the later $\effR(w_i)$).
            $T[p..] \bowtie \effR(v_{c(i,j)}) \bowtie B_{i,j}$ starts with $\Filled^3 \Empty \Filled \overline{w_i} w_i$.
            However, we have no string suitable for filling the leftmost wildcards.
    \end{description}

    Case 4 considers $Q[p..] = \overline{w_i} w_i \Filled \effR(v_{c(i,j}) \circ \cdots$.
    We need a string starting with $\effL(w_i)$ or $\effR(w_i)$.
    For the latter, no such string exists.
    For the former, we have two choices:
    the patch string $\effL(w_i)$ itself, or $B_{i,j}$.
    \begin{description}
        \item[4-A]:
            $Q[p..] \bowtie \effL(w_i) = \Filled^3 \effR(v_{c(i,j}) \circ \cdots$, a setting of Case~2'.
        \item[4-B]:
            $Q[p..] \bowtie B_{i,j} = \Filled^3 \overline{v_{c(i,j)}} v_{c(i,j)} \Filled \effR(w_i) \circ \cdots$, a setting of Case~4'.
    \end{description}

    Case 4' considers $Q[p..] = \overline{v_{c(i,j)}} v_{c(i,j)} \Filled \effR(w_i) \circ \cdots$.
    We need a string starting with $\effL(v_{c(i,j)})$ or $\effR(v_{c(i,j)})$.
    For the former, we have $B_{i,j}$, and for the latter the patch string $\effR(v_{c(i,j)})$ itself.
    \begin{description}
        \item[4'-A]:
            $Q[p..] \bowtie B_{i,j} = \Filled^3 \overline{w_i} w_i \Filled \effR(v_{c(i,j+1)}) \circ \cdots$, a setting of Case~4.
        \item[4'-B]:
            $Q[p..] \bowtie \effR(v_{c(i,j)}) = \Filled^3 \effR(w_i) \circ \cdots$, a setting of Case~2.
            This case is prohibitive.
    \end{description}

    Of all the cases we can encounter,
    Cases~1-H-2, 2'-B, 3-B-1, 3-B-2, and 4'-B are prohibitive.
    However, we claim that any transition that leads to a prohibitive case cannot lead to constructing $Q$, therefore showing the lemma.

    \begin{description}
        \item[1-H-2]:
        By selecting Case~1-H-2, we have used the patch string $\effL(w_i)$, leading us to Case~3', where we can only continue with Case~3'-A to Case 4.
        Here, we cannot select Case~4-A since we have already used $\effL(w_i)$ in Case~1-H-2.
        Thus, we can only continue with Case 4-B, which leads to case 4'.
        We prove later that using case 4'-B cannot be used to construct $Q$.
        The remaining case 4'-A leads to case 4,
        and thus again to case 4-B, and we cannot escape this cycle unless we use case 4'-B.

\item[2'-B]:
        Case~2'-B makes the choice to apply $T[p..] \bowtie \effM(v_{c(i,j)})$ and leads to Case~1.
        Note that $v_{c(i,j)} \neq v_1$ because $\effM(v_1)$ is not in $\mathcal{T}$.
        The problem with using $\effM(v_{c(i,j)})$ is that all strings $C_i$ of $\mathcal{C}$ with $i \ge 2$ appear only in Cases~1-B, 1-C, 2'-C, and 2'-D, and each of them requires a matching with $\effM(v_{c(i,j)}$.
        Thus, we end up with wildcards that we cannot fill.

        \item[3]:
        Case 3 can only be reached from case 1-F or case 2'-A (cf.~Fig.~\ref{figure:transitionGraph}).
        Both Cases~1-F and~2'-A use $\effM(v_{c(i,j)})$ for some $i,j$,
        and thus with the same argument for Case~2'-B,
        it follows that these transitions make it impossible to use a Case-1 transition for adding $C_i$ to $Q$.
        Thus the built $Q$ either contains wildcards or no $C_i$.
        Thus, Case 3 (and all its subcases) must not be used for building $Q$.
A consequence is that also Cases~1-F and~2'-A leading to Case~3 must not be used.

        \item[4'-B]:
        To reach 4'-B, from Case 1, the last transitions must be
        \begin{itemize}
            \item $\cdots $2'-A$ \to $ 3-B-2$ \to $ 4-B$ \to $ 4',
            \item $\cdots $1-F$ \to $ 3$ \to $ 3-B-2$ \to $ 4-B$ \to $ 4',
            \item $\cdots $2-B$ \to $ 3'-A$ \to $ 4-B$ \to $ 4', or
            \item $\cdots $1-H-2$ \to $ 3'-A$ \to $ 4-B$ \to $ 4'.
        \end{itemize}
        So any sequence of transitions must go through Cases~1-F, 1-H-2, 2-B, or 2'-A.
        On the one hand, as stated above, no transition to Case~3 must be used, so Cases~1-F and~2'-A cannot be used.

        On the other hand, the remaining Cases~1-H-2 and~2-B use the patch string $\effM(w_i)$.
        However, Case~4'-B leads us to Case~2, and all its subcases require the already used $\effM(w_i)$ to proceed.
        It follows that using Case~4'-B leads to a dead end, since we cannot access it from Case~1, from which we initially start.
    \end{description}
\end{proof}

From Lemma~\ref{lemma:block_aligned}, it follows that any superstring of $\mathcal{T}'$ of length $(2m+3n)3\ell'$ must be block-aligned,
implying that the strings $\{ f(xyz) \mid xyz\in \mathcal{T}\}$ can only overlap in the same way as $xyz$ can overlap, implying a superstring of length $2m+3n$ of $\mathcal{T}'$, and thus a Hamiltonian path from $s$ to $t$.

\begin{figure}
\makebox[\textwidth]{\centerline{\includegraphics*[page=1,clip=true,width=0.5\textwidth,frame]{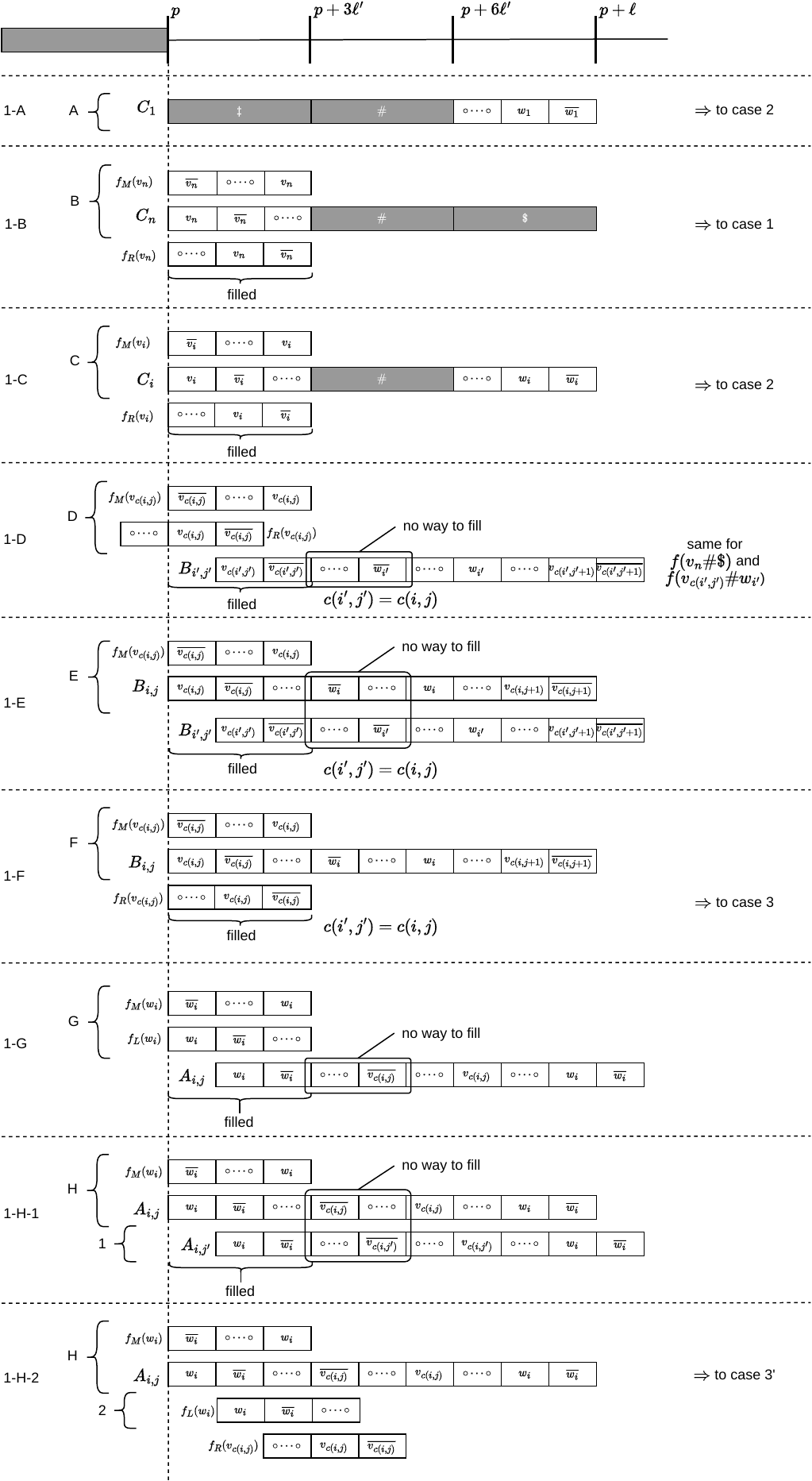}
        \includegraphics*[page=3,clip=true,width=0.7\textwidth,frame]{img/cases.pdf}
    }
}

\begin{minipage}{0.3\linewidth}
    \centering $\uparrow$ Case 1,\hspace{3em} $\downarrow$ Case 2
\end{minipage}
\begin{minipage}{0.65\linewidth}
    \centering $\uparrow$ Case 3,\hspace{6em} $\downarrow$ Case 4
\end{minipage}

\makebox[\textwidth]{\centerline{\includegraphics*[page=2,clip=true,width=0.5\textwidth,frame]{img/cases.pdf}
        \includegraphics*[page=4,clip=true,width=0.7\textwidth,frame]{img/cases.pdf}
    }

}\caption{Cases considered in the proof of Lemma~\ref{lemma:block_aligned}.}
\label{figBlockAlignedCompact}
\end{figure}

\JFigure{\begin{figure}[htbp]
    \centerline{\includegraphics*[page=1,clip=true,width=0.9\textwidth,frame]{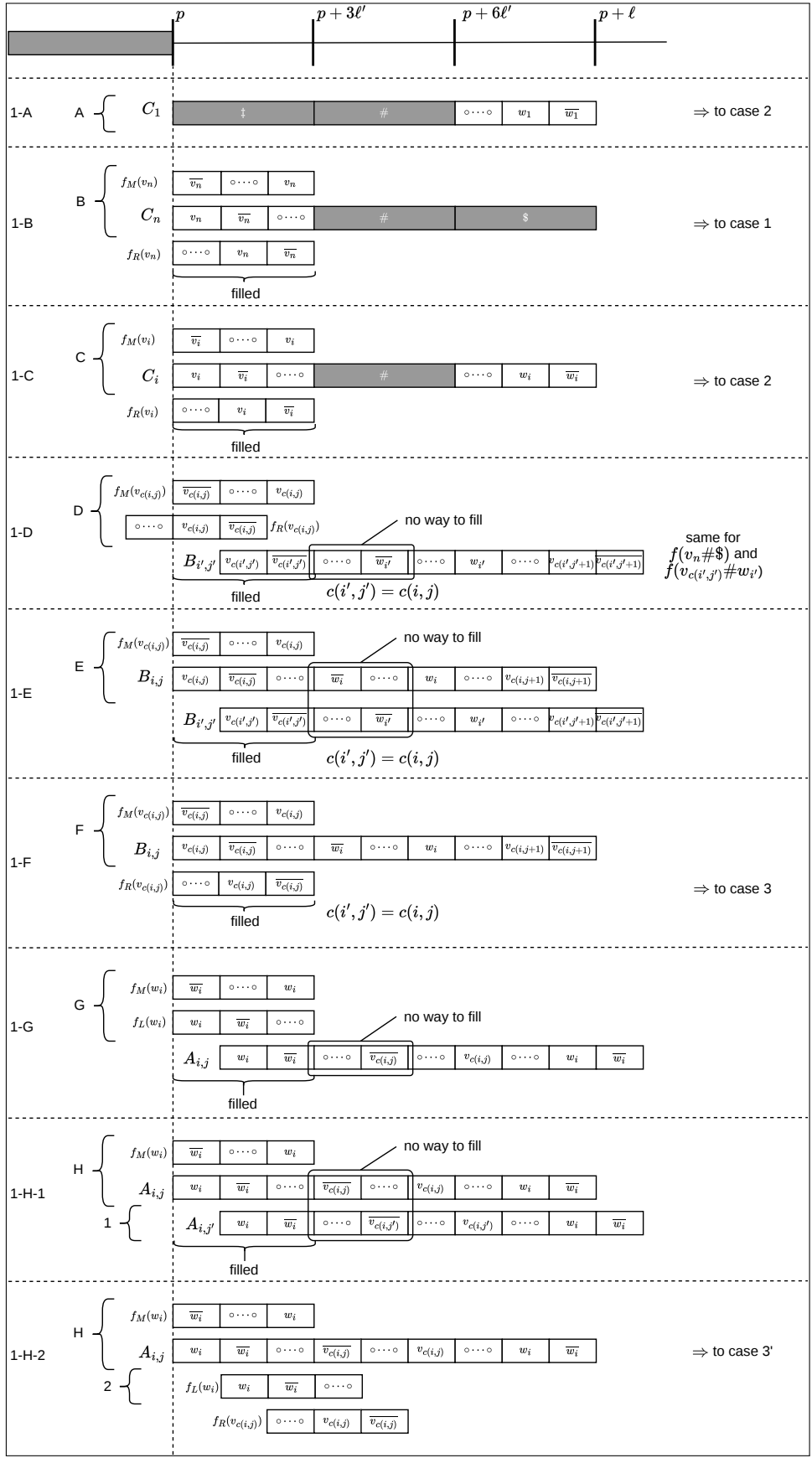}
    }
    \caption{Case 1 for proving Lemma~\ref{lemma:block_aligned}}\label{figure:block_aligned_rem0}
\end{figure}
}\JFigure{\begin{figure}[htbp]
    \centerline{\includegraphics*[page=2,clip=true,width=0.9\textwidth,frame]{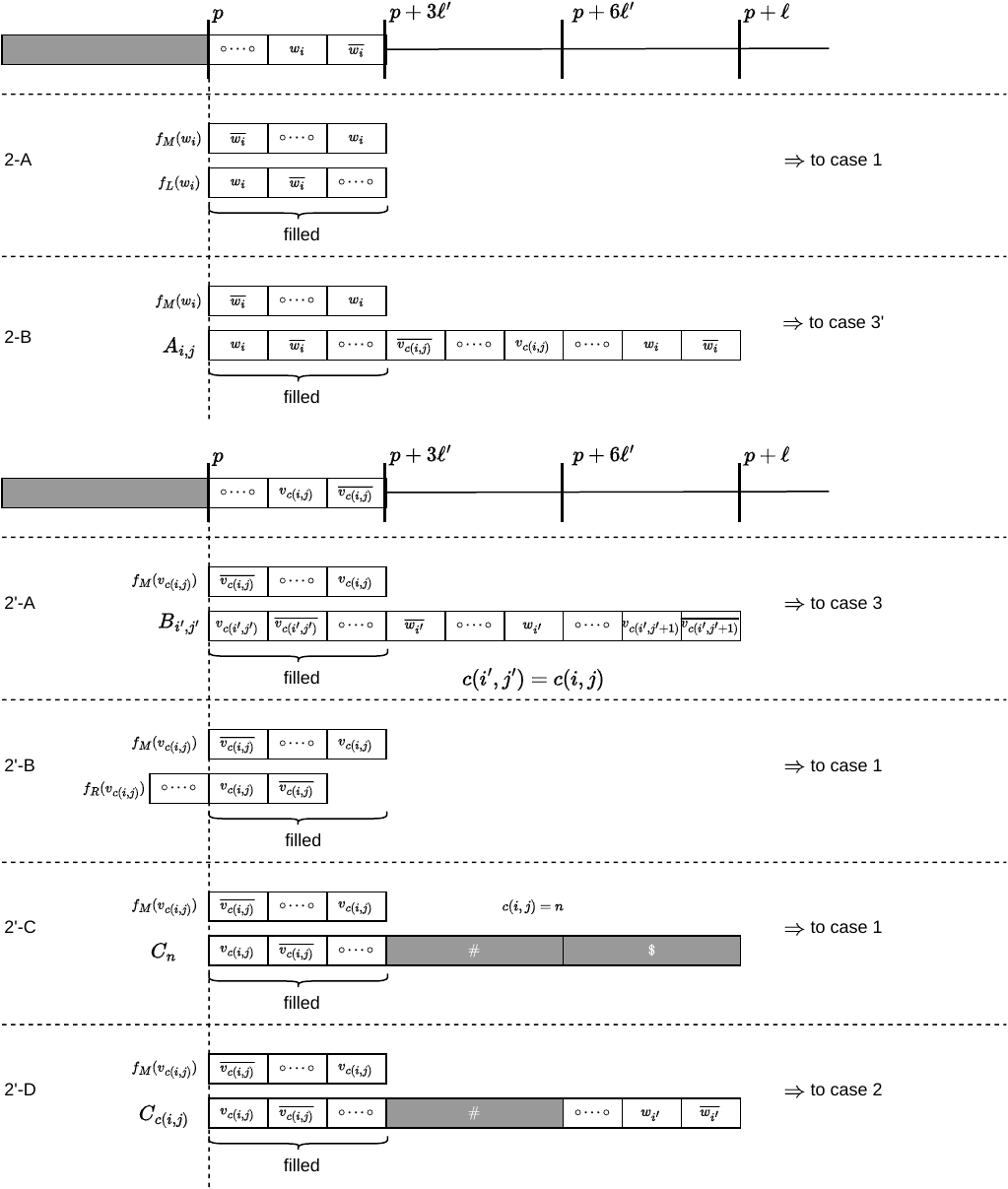}
    }
    \caption{Case 2 for proving Lemma~\ref{lemma:block_aligned}}\label{figure:block_aligned_rem1}
\end{figure}
}\JFigure{\begin{figure}[htbp]
    \centerline{\includegraphics*[page=3,clip=true,width=0.9\textwidth,frame]{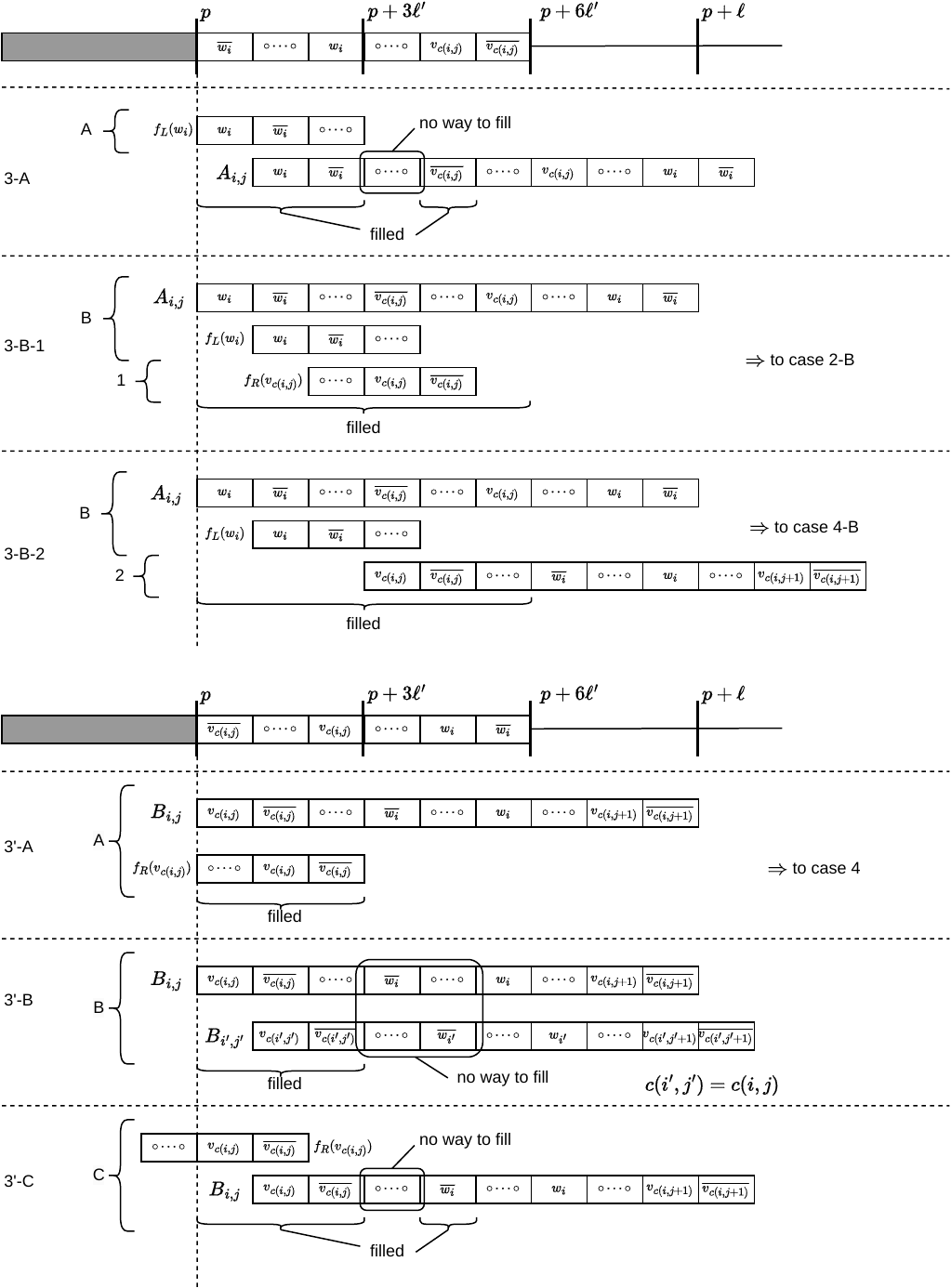}
    }
    \caption{Cases 3 and 3' for proving Lemma~\ref{lemma:block_aligned}}\label{figure:block_aligned_rem2}
\end{figure}
}\JFigure{\begin{figure}[htbp]
    \centerline{\includegraphics*[page=4,clip=true,width=0.9\textwidth,frame]{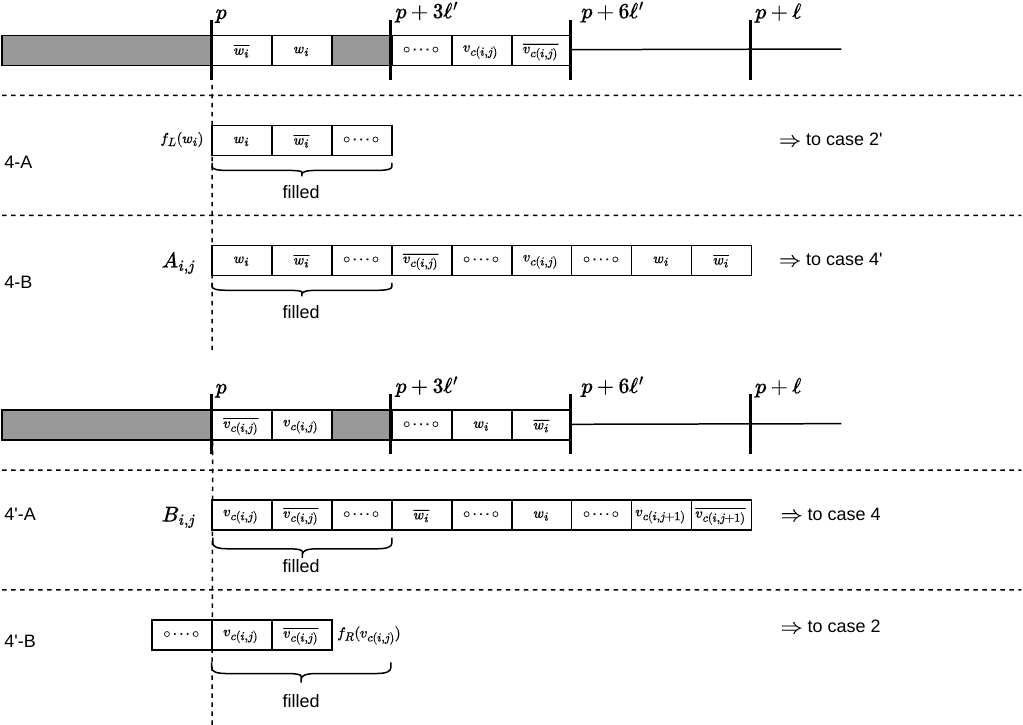}
    }
    \caption{Cases 4 and 4' for proving Lemma~\ref{lemma:block_aligned}}\label{figure:block_aligned_rem21}
\end{figure}
}

\clearpage

\appendix
\JO[\section{Reduction from 3-Colorability}]{}

\Jvar{}

\JO[\section{Supplementary Figures}]{}
\Jfigs{}

\clearpage
\section{MAX-SAT formulation}
\newcommand{\nbase}[1]{b_{#1}}
\newcommand{\ncheck}[1]{c_{#1}}
\newcommand{\nodemax}{n}
\newcommand{\nodeinternal}{V_I}
\newcommand{\posmax}{N}
\newcommand{\lblset}[1]{L_{#1}}
\newcommand{\lbl}[1]{l_{#1}}
\newcommand{\pad}[1]{p_{#1}}
\newcommand{\charmap}[1]{m_{#1}}
\newcommand{\charmat}[1]{g_{#1}}
\newcommand{\charmax}{C}

\newcommand{\distcost}[1]{d_{#1}}

The minimization problem can be formulated as a MAX-SAT instance.
MAX-SAT is a variant of SAT,
where we are given two sets of CNF clauses, hard and soft,
and the output is a truth assignment of the variables in which
all hard clauses are satisfied and the number of satisfied soft clauses is maximized~\cite{SATHandbook}.

Let $\posmax$ be an upper bound on the maximum size of base/check,
and let
the set of outgoing edge labels from node $i$ be denoted as
$\lblset{i}=\{a_1, \ldots, a_{|\lblset{i}|}\} \subseteq [1, \sigma]$

We define the following $O(\nodemax \posmax)$ Boolean variables:
\begin{itemize}
    \item $\nbase{i,j}$ for $i\in V$, $j \in [1, \posmax - \max(\lblset{i})]$: $\nbase{i,j}=1$ iff $\barr[\node{i}]=j$,
    \item $\ncheck{i,j}$ for
          $i \in [1,\nodemax]$, $j\in[1,\posmax]$: $\ncheck{i,j} = 1$ iff
          $\exists (i, \cdot, k) \in E$ s.t. $\carr[\node{k}=j]=\node{i}$,
    \item $\pad{j}$ for $j \in [1, \posmax]$: $\pad{j}=0$ iff $\carr[j^\prime]= \mathit{None}$ $\forall j^\prime \geq j$. \end{itemize}
We first give the hard clauses, which define the constraints these variables must satisfy.

The relationship between base and check values satisfy:
\begin{gather}
    \forall i\in V, \forall j\in[1,\posmax-\max(\lblset{i})], \forall a \in \lblset{i},
    \nbase{i,j}\rightarrow \ncheck{i,j+a}
\end{gather}

Each node $i$ can only have one base value:
\begin{gather}
    \forall i\in V, \sum_{j \in [1,\posmax-\max(\lblset{i})]} \nbase{i,j}= 1
\end{gather}

Each position $j$ only holds at most one check value:
\begin{gather}
    \forall j \in [1,\posmax], \sum_{i\in[1,\nodemax]} \ncheck{i, j}\leq 1
\end{gather}

Finally, $\pad{i}$ satisfies:

\begin{gather}
    \forall i \in [1,\nodemax], \forall j \in [1,\posmax], \ncheck{i,j}\rightarrow \pad{j} \\
    \forall j \in [2,\posmax], \pad{j} \rightarrow \pad{j-1}
\end{gather}

Next, for the soft clauses, we consider the following:

\begin{gather}
    \forall i \in [1, N], \lnot \pad{i}
\end{gather}

By the definition of $\pad{i}$,
$\sum_{j\in[1,\posmax]} \lnot \pad{j}$
is the size of the double-array, and thus maximizing the number of satisfied soft clauses minimizes the size of the double-array.

\subsection{Computational Experiments}

We implemented the MAX-SAT formulation using
the PySAT library~\footnote{\url{https://pysathq.github.io/}}.
Preliminary experiments on a small trie of 70 nodes and setting $N=256$,
the computation took more than 24 hours.
We also implement finding a semi-optimal value with the following strategies,
and measured the computational time for different sizes of tries.
(1) Discard the soft constraints, and solve the SAT problem defined by the hard constraints, to see if there exists a double-array of size $N$,
(2) set a time-out
(3) compute the smallest double-array that can be found within the time-out via binary search .

The method is tested on the words dataset and the first $x = 30, 100, 200, 300$ words are represented in the trie. Table~\ref{tab:experiment} summarizes the results.
Here, \#greedy and \#sat respectively denote the found double-array size with the greedy algorithm and semi-optimal algorithm.
Density is the size of the trie with respect to the double-array size, i.e., the number
of valid elements in the double-array.
"-" denotes that the computation did not finish in 1 hour.
We can see that SAT always finds a smaller double-array than greedy, especially for small trie sizes.

\begin{table}[ht]
    \tiny
\centering
    \caption{Comparison between greedy and semi-optimal}
    \begin{tabular}{|l|r|r|r|r|r|r|}
        \hline
        ~          & \# trie nodes & \#greedy & greedy density & \#SAT & SAT density & time-out (min) \\ \hline
        word\_30   & 293           & 614      & 0.48           & 389   & 0.75        & 30             \\ \hline
        word\_100  & 919           & 1126     & 0.81           & 1015  & 0.90        & 40             \\ \hline
        words\_200 & 1792          & 1907     & 0.94           & 1907  & 0.94        & 40             \\ \hline
        words\_300 & 2598          & 2918     & 0.89           & -     & -           & -              \\ \hline
    \end{tabular}
    \label{tab:experiment}
\end{table}

\end{document}